\documentclass[preprint,12pt]{elsarticle}

\usepackage{amssymb}
\usepackage{amsmath}
\usepackage{amsthm}
\usepackage{multirow}
\newtheorem{proposition}{Proposition}
\newtheorem{assumption}{Assumption}
\newtheorem{condition}{Condition}
\newtheorem{theorem}{Theorem}
\usepackage{bm}
\usepackage{float}
\usepackage{subcaption}
\usepackage[a4paper, margin=1in]{geometry}
\journal{arXiv}

\begin{document}

\begin{frontmatter}

%% Title, authors and addresses

%% use the tnoteref command within \title for footnotes;
%% use the tnotetext command for theassociated footnote;
%% use the fnref command within \author or \address for footnotes;
%% use the fntext command for theassociated footnote;
%% use the corref command within \author for corresponding author footnotes;
%% use the cortext command for theassociated footnote;
%% use the ead command for the email address,
%% and the form \ead[url] for the home page:
%% \title{Title\tnoteref{label1}}
%% \tnotetext[label1]{}
%% \author{Name\corref{cor1}\fnref{label2}}
%% \ead{email address}
%% \ead[url]{home page}
%% \fntext[label2]{}
%% \cortext[cor1]{}
%% \affiliation{organization={},
%%             addressline={},
%%             city={},
%%             postcode={},
%%             state={},
%%             country={}}
%% \fntext[label3]{}

\title{Real-time Bus Travel Time Prediction and Reliability Quantification: A Hybrid
Markov Model}

%% use optional labels to link authors explicitly to addresses:
%% \author[label1,label2]{}
%% \affiliation[label1]{organization={},
%%             addressline={},
%%             city={},
%%             postcode={},
%%             state={},
%%             country={}}
%%
%% \affiliation[label2]{organization={},
%%             addressline={},
%%             city={},
%%             postcode={},
%%             state={},
%%             country={}}

\author[inst1]{Yuran Sun\corref{cor1}}
\ead{yuransun@ufl.edu}
\cortext[cor1]{Corresponding author.}

\affiliation[inst1]{organization={Department of Civil and Coastal Engineering},%Department and Organization
            addressline={University of Florida}, 
            city={Gainesville},
            postcode={32608}, 
            state={FL},
            country={US}}

\author[inst2]{James C. Spall}
\author[inst3]{Wai Wong}
\author[inst1]{Xilei Zhao}

\affiliation[inst2]{organization={Department of Applied Mathematics and Statistics},%Department and Organization
            addressline={John Hopkins University}, 
            city={Baltimore},
            postcode={21211}, 
            state={MD},
            country={US}}
\affiliation[inst3]{organization={Department of Civil and Natural Resources Engineering},%Department and Organization
            addressline={University of Canterbury}, 
            city={Christchruch},
            postcode={8140}, 
            country={New Zealand}}

\begin{abstract}
%% Text of abstract
Accurate and reliable bus travel time prediction in real-time is essential for improving the operational efficiency of public transportation systems. However, this remains a challenging task due to the limitations of existing models and data sources. This study proposed a hybrid Markovian framework for real-time bus travel time prediction, incorporating uncertainty quantification. Firstly, the bus link travel time distributions were modeled by integrating various influential factors while explicitly accounting for heteroscedasticity. Particularly, the parameters of the distributions were estimated using Maximum Likelihood Estimation, and the Fisher Information Matrix was then employed to calculate the 95\% uncertainty bounds for the estimated parameters, ensuring a robust and reliable quantification of prediction uncertainty of bus link travel times. Secondly, a Markovian framework with transition probabilities based on previously predicted bus link travel times was developed to predict travel times and their uncertainties from a current location to any future stop along the route. The framework was evaluated using the General Transit Feed Specification (GTFS) Static and Realtime data collected in 2023 from Gainesville, Florida. The results showed that the proposed model consistently achieved better prediction performance compared to the selected baseline approaches (including historical mean, statistical and AI-based models) while providing narrower uncertainty bounds. The model also demonstrated high interpretability, as the estimated coefficients provided insights into how different factors influencing bus travel times across links with varying characteristics. These findings suggest that the model could serve as a valuable tool for transit system performance evaluation and real-time trip planning.

\end{abstract}

\end{frontmatter}

%% \linenumbers

%% main text
\section{Introduction}
Public transportation is a cornerstone of urban mobility, connecting millions to key destinations \citep{APTA_Public_Transportation_Fact}. Buses, in particular, play a crucial role in alleviating traffic congestion \citep{he2021improving} and emissions due to their high passenger capacity per unit of road space, positioning them as a sustainable transportation option \citep{kumar2017bus}. However, their attractiveness to commuters is often hampered by the high uncertainty of travel time, compared to private cars \citep{ma2019bus, bai2015dynamic, APTA_Transit_Ridership_2024}. Accurate travel time prediction in real-time is essential to enhance service quality and passenger experience and thus increase bus ridership \citep{mori2015review, estrada2015experimental}. Nonetheless, achieving this remains a challenging task due to unpredictable traffic conditions, signal delays, and other stochastic factors affecting bus operations \citep{he2018travel}.

To address the challenges in bus travel time prediction, many studies have explored a wide range of modeling approaches \citep{bai2015dynamic, petersen2019multi, kumar2019real, amita2015prediction, yu2018prediction}. Among them, traditional statistical models are widely applied due to their well-established theoretical foundations and explicit mathematical formulations, which enhance interpretability and provide insights into the relationships between travel time and influencing factors \citep{rudin2019stop, erasmus2021interpretability}. Learning these relationships is crucial for transit agencies and urban planners, as it enables data-driven interventions to optimize transit operations \citep{mohamed2021identification}. For instance, understanding how peak-hour demand affects travel time can guide the implementation of bus-priority signal control at critical intersections, reducing delays and improving service reliability \citep{mirchandani2001approach}. Additionally, in recent years, AI-based models (e.g., Random Forest, Multilayer Perceptron) have demonstrated superior predictive capabilities and have become integral to the development of real-time bus travel time information systems \citep{huang2021bus, singh2022review}. While existing studies have contributed significantly to bus travel time prediction, major limitations remain, particularly in terms of modeling methodologies and data availability.

Firstly, while statistical models provide greater interpretability, they often lag behind AI-based models in terms of predictive accuracy \citep{yu2018prediction, ashwini2022bus}. On the other hand, although AI-based models exhibit superior predictive performance in bus travel time forecasting, their black-box nature raises concerns about their ability to elucidate the underlying relationships between influencing factors and travel time \citep{rudin2019stop, erasmus2021interpretability}. Consequently, developing a modeling framework that balances interpretability with predictive performance, comparable to that of AI-based models, remains a crucial area of ongoing research.

The second limitation of current bus travel time prediction models is their exclusive focus on point estimates, which do not adequately capture travel time reliability. Studies addressing the reliability of predictions (that consider potential uncertainties and variations in travel times) remains scarce \citep{chen2024conditional, o2016uncertainty, yetiskul2012public}. Reliability, nevertheless, is equally crucial for enhancing bus service quality \citep{yetiskul2012public, ma2015modeling, yu2017using}, as it allows passengers to plan more effectively by accounting for travel time variations and enables transit agencies to monitor system performance fluctuations and optimize operations accordingly. Existing studies on travel time forecasting and uncertainty estimation have employed probabilistic frameworks, such as Markovian \citep{ma2017estimation, buchel2022modeling} and Bayesian approaches \citep{huang2021bus, chen2023probabilistic, chen2024conditional} to effectively model the stochastic nature of bus travel times. Yet, these methods still have notable limitations. Most probabilistic models rely solely on the inherent characteristics or statistical distributions of bus travel times, neglecting the influence of external factors such as weather and traffic conditions. Furthermore, these models are often static in nature \citep{chen2023probabilistic}, limiting their effectiveness in real-time and dynamic prediction.

Another limitation of existing models is that their performance is typically evaluated using empirical data \citep{yu2018prediction, moosavi2023evaluation}, with limited theoretical validation. In particular, there is a lack of consistency analysis, which examines whether a model’s estimates converge to the true underlying values as the sample size increases. Proving consistency of model parameters is crucial, as it ensures that the model’s performance is reliable given enough data.

In addition to modeling limitations, existing studies also face challenges related to data availability. Most research relies on Automatic Vehicle Location (AVL) data \citep{stephentransit} and GPS-based bus trajectory data \citep{taparia2021bus}, which provide rich information on bus locations and movements \citep{wu2023gtfs}. However, these datasets are often not publicly accessible, limiting their widespread use. Furthermore, they are not sufficiently comprehensive for accurate travel time prediction, as they lack critical information such as bus stop locations and real-time traffic conditions, necessitating the integration of supplementary data sources \citep{ma2019bus}. The General Transit Feed Specification (GTFS) data has recently emerged as a valuable alternative, offering publicly available static and real-time data on bus operations. GTFS static data provides essential information, including stop locations, route geometries, and schedules, while GTFS Realtime delivers standardized updates on real-time vehicle locations and service alerts. These features make GTFS a promising resource for monitoring bus movements and predicting travel times. Despite its considerable potential, GTFS remains relatively under-explored in comparison to other location data sources for bus travel time prediction. Although a few studies have leveraged GTFS to infer travel time \citep{chondrodima2021public, bv2024travel, ou2022ai}, its broader applications in this field are still in the early stages.

To address these research gaps, this study proposes a hybrid methodological framework for bus travel time prediction and uncertainty quantification, consisting of two key parts. First, the study models the distributions of three primary components of bus travel time—road travel time, stop dwell time, and intersection waiting/passing time—separately for each bus link, defined as the segment between two consecutive bus stops. Given that road travel time, the time a bus spends actively moving on the road, is highly affected by traffic conditions and congestion patterns, this study specifically focuses on modeling road travel time while accounting for these influencing factors. The parameters of the distributions are estimated using Maximum Likelihood Estimation (MLE) and the corresponding uncertainty bounds are derived from the Fisher Information Matrix (FIM). These distributions are then used to predict bus travel times for each link, incorporating the estimated uncertainty bounds. In the second part, a Markovian framework leveraging transition probabilities based on previously predicted bus link travel times is introduced to predict real-time travel times and estimate uncertainty bounds for trips from a current location to any future stop along the route (as captured through periodic data updates). The proposed framework is implemented and validated using GTFS static and real-time data. Its performance is evaluated against various baseline models, including historical mean, statistical, and AI-based approaches, to assess its accuracy and reliability. The results indicate that the proposed model consistently achieves satisfying performance in bus travel time prediction, demonstrating predictive accuracy comparable to AI-based techniques. Additionally, it generates narrower uncertainty bounds while offering greater interpretability, making it a more transparent and reliable alternative for bus travel time prediction.

The proposed study offers several key contributions: First, it proposes a novel methodological framework that integrates the MLE process with a Markov chain model. This approach not only achieves real-time travel time predictions with predictive performance comparable to that of AI-based models but also provides narrower uncertainty bounds than those generated by traditional statistical models, demonstrating greater confidence and precision in its predictions.

Secondly, by incorporating parametric formulations that account for the factors influencing travel time distributions, the proposed model enhances the understanding of how these factors impact travel times. Its high interpretability facilitates informed decision-making by stakeholders.

Third, the model makes a theoretical contribution through the rigorous proof of its consistency. Such a foundation not only reinforces the model's credibility under conditions of sufficient data availability but also establishes a robust basis for its practical application in real-world scenarios.

Finally, the model is particularly well-suited for the new publicly available GTFS data. Specific algorithms have been developed to infer time components based on the unique structure of GTFS data, and the model is designed to seamlessly incorporate periodic updates from GTFS Realtime feeds. Additionally, the method enables the automatic detection of traffic conditions by directly analyzing bus movement patterns using GTFS Realtime data. This eliminates the reliance on additional, non-public data sources.

The structure of the paper is as follows: Section 2 provides a review of the literature on existing methods for bus travel time prediction, modeling bus travel time distributions, and probabilistic forecasting approaches. Section 3 presents the methodological framework, including the inference of travel time components, the prediction of individual link travel time components, the estimation of their uncertainty bounds, and the development of the final Markovian framework for real-time bus travel time prediction. Section 4 introduces the bus routes and datasets used in this study, and reports the results of statistical tests evaluating the assumptions proposed in the methodology. The results and corresponding discussions are presented in Section 5. Finally, Section 6 summarizes the research and key findings, explores potential applications, discusses limitations, and outlines directions for future research.

\section{Literature Review}
\subsection{Travel Time Prediction}
Bus travel time prediction has garnered significant attention in recent decades due to its pivotal role in facilitating users to make informed pre-trip and en-route decisions, enhancing their satisfaction, and hence boosting ridership \citep{mori2015review}. Many methodologies have been developed to address the complexities of travel time forecasting. These methods can be broadly categorized based on the types of models employed, ranging from historical mean approaches to statistical models, and, more recently, AI-based techniques.

In earlier years, the historical mean method \citep{4357735}, which predicts travel time by averaging previously recorded travel times, was widely adopted due to its simplicity and computational efficiency \citep{114368}. However, this approach fails to account for the stochastic nature of traffic conditions, which can result in reduced reliability in real-world applications.

Statistical models used for travel time prediction typically follow a predetermined mathematical model and require the estimation of key parameters \citep{mori2015review}. These parameters are often estimated based on historical data to optimize model performance. Common examples of statistical models include regression models, time-series models, and the Kalman Filter. Regression models often underperform compared to more complex approaches, primarily due to their strict assumptions and limited capacity to capture the complex relationships between influencing factors and bus travel times. Consequently, they are commonly utilized as baseline models \citep{yu2018prediction}. For time-series models, autoregressive integrated moving average (ARIMA) models are commonly applied for predicting bus travel times. For example, \citet{suwardo2010arima} employed an ARIMA model to forecast bus travel times between end-to-end bus terminals. In addition, autoregressive (AR) models and their variants are widely applied. \citet{dhivya2020bus} developed two AR-based models for predicting bus travel times across 500-meter route segments: a seasonal AR model to address non-stationary travel patterns and a linear non-stationary AR model leveraging temporal correlations in travel time data. Kalman Filter approaches are particularly effective in dynamic environments requiring real-time updates, as they continuously adjust predictions based on new data \citep{aldokhayel2018kalman}. Moreover, Kalman Filters often perform well when integrated with other predictive models, such as time-series \citep{xu2017real} or AI-based models \citep{bai2015dynamic}. \citet{achar2019bus} proposed a spatial Kalman Filter model that captures the spatial and temporal correlations in traffic to facilitate real-time predictions of bus travel times across 500-meter sections. These statistical methods are often grounded in well-established theories and are considered highly interpretable due to their parametric formulations or probabilistic frameworks. Additionally, their rigorous structure inherently facilitates the quantification of uncertainty in bus travel times. However, their predictive power may be limited compared to AI-based models.

AI-based models are often equipped to learn more complex effects of various factors on bus travel time \citep{chen2023probabilistic}, typically resulting in high prediction accuracy. \citet{ashwini2022bus} conducted a comparative study using bus GPS data between statistical models, such as linear and ridge regression, and AI-based models, demonstrating that the latter significantly outperformed the former in predicting bus travel times of some selected entire bus routes in the Tumakuru Smart City case study. Basic AI techniques, such as decision trees \citep{moosavi2023evaluation}, random forests \citep{yu2018prediction}, boosting algorithms \citep{zhang2015gradient} and support vector machines (SVM) \citep{ma2019bus}, are commonly employed in many existing studies. These models are relatively straightforward to implement and require less in feature dimension. For example, \citet{yu2018prediction} proposed a Random Forest model integrated with a nearest neighbor approach to predict bus travel times between stops using AVL data, achieving high prediction accuracy and strong model performance. With the advancement of modern technologies, many studies have adopted more sophisticated models, such as deep neural networks (DNNs) \citep{wang2018will} and Long Short-Term Memory (LSTM) networks \citep{petersen2019multi}, to predict travel time. These models excel at automatically extracting complex patterns from data, making them particularly suited for large-scale and intricate prediction tasks. Among these, Artificial Neural Networks (ANNs) remain one of the most widely utilized models in the field of travel time prediction \citep{chien2002dynamic, jeong2004bus, gurmu2014artificial, ramakrishna2015bus, amita2016prediction}. \citet{as2018dynamic} demonstrated that ANNs significantly outperformed historical mean-based methods in predicting bus travel times between adjacent stops using bus probe data. Additionally, \citet{petersen2019multi} integrated LSTM networks with convolutional layers, which are typically used for graph-structured data, to predict bus link travel times. This hybrid approach enabled the model to capture both spatial and temporal correlations within transit networks, thereby improving its overall predictive performance. Although AI-based models typically achieve high predictive performance, they are often considered black-box models due to their lack of interpretability. This makes it challenging to directly understand the impact of key factors on bus travel times. Moreover, these models generally require large-scale training data with high temporal granularity—particularly temporal AI-based models—which may not be suitable for analyzing travel times on low-frequency bus routes.

Although extensive research has been conducted on bus travel time prediction, and statistical models are well-suited for uncertainty quantification, most existing methods focus on point estimation, providing a single expected travel time. Approaches that explicitly account for variability in their predictions remain underexplored in the literature.

\subsection{Bus Travel Time Distribution and Probabilistic Forecasting}
This subsection focuses on the distributional characteristics and probabilistic forecasting of bus travel times, highlighting the unique patterns in bus travel time distributions compared to normal traffic.

Modeling the distribution of bus travel times is essential for accurate probabilistic forecasting, as it enable estimation of bus travel times and their uncertainties across varying conditions. Early studies in this area typically modeled the distribution of daily bus travel times and average trip duration. \citet{taylor1982travel} found that daily bus travel time typically follows normal distribution, while \citet{jordan1979zone} identified that average trip times conform to gamma distribution. Research by \citet{mazloumi2010using} investigated travel time distributions across different departure time windows, revealing that travel times closely adhere to a normal distribution during narrow departure windows and peak hours within broader windows. Conversely, in non-peak hours within broad windows, travel times often follow a log-normal distribution. Numerous studies have supported the finding that bus travel time tends to be right-skewed and aligns with a log-normal distribution \citep{uno2009using, duran2016estimation, kieu2015public}. Comparative analyses of various distributions, such as those by \citet{chepuri2020development} and \citet{harsha2022probability}, suggest that bus travel times are more aptly described by the generalized extreme value (GEV) distribution. However, \citet{ma2016modeling} found that the performance of common unimodal distributions, such as normal, log-normal, and GEV, in modeling travel time shows negligible differences. Notably, when modeled as a multi-modal distribution, bus travel times often conform to conform to a Gaussian mixture distribution. This conclusion was reached by \citet{ma2016modeling}, who conducted the Anderson–Darling (AD) test to evaluate and select distribution candidates, and compared the fitting errors, and metrics of symmetry, normality, and multimodality of bus travel times across Gaussian mixture distributions and other candidate distributions. To model bus travel time distributions, statistical frameworks are frequently employed, with Markovian frameworks being particularly common for estimating these distributions. For instance, \citet{ma2017estimation} applied a generalized Markov chain approach to estimate trip travel time distributions based on AVL data, formulating transition probabilities as functions of explanatory variables. \citet{buchel2022modeling} employed a hidden Markov chain approach to estimate road travel time and stop dwell time distributions using open-source departure and arrival time data, while accounting for the conditional dependence between section travel times.

Current literature predominantly employs Bayesian frameworks for probabilistic forecasting models. For example, \citet{huang2021bus} introduced a Bayesian Support Vector Regression framework integrated with Functional Data Analysis (FDA) to predict short-term link travel times using GPS data while accounting for various uncertainties. Similarly, \citet{chen2023probabilistic} proposed a Bayesian Gaussian mixture model to forecast bus link travel time and associated uncertainties with bus in-out-stop record data. Their model utilized a Gaussian mixture distribution, assuming consistent Gaussian components throughout different times of the day, with time-dependent mixing coefficients to capture dynamic changes in bus operations. Markov chain Monte Carlo methods were used to derive the posterior distributions of the parameters, enabling probabilistic forecasting after parameter estimation. \citet{chen2024conditional} subsequently introduced a novel Bayesian Markov regime-switching vector regression model, which demonstrated improved performance compared to the earlier Bayesian Gaussian mixture model. 

These forecasting models contribute to advancements in bus link travel time prediction and uncertainty estimation. However, relatively few studies have explored the prediction of remaining travel time from a given location to any future bus stop along the route. Additionally, it remains unclear whether some of these models can achieve predictive performance comparable to existing baseline models, especially AI-based models.

\section{Method}
A bus link is defined as the road segment between two adjacent bus stops. The bus link travel time consists of three main components: road travel time, stop dwell time, and intersection waiting and passing time, each of which is modeled separately. These relationships can be expressed through the following equation:
\begin{equation}
    \Delta t_{l} = \Delta t_{r} + \Delta t_{d} + \Delta t_{s},
\end{equation} 
where $\Delta t_l$ represents the total link travel time, $\Delta t_r$ is the road travel time, $\Delta t_d$ denotes the stop dwell time, and $\Delta t_s$ accounts for the intersection waiting and passing time.

The process of the methodology consists of three main steps: First, the three types of time are inferred. Next, their distributions are modeled, and both their predicted values and associated uncertainty bounds are estimated. Finally, a Markov chain model is introduced to generate real-time predictions of bus travel times from the current location to any future stop along the route.

\subsection{Link Travel Time Inference}
The three types of travel times are inferred using a combination of GTFS static and Realtime data, as well as road segment information from OpenStreetMap. Specifically, the bus route's shape within the road network, the stop locations, and the length of each bus link are extracted from the static GTFS data. Intersection locations are obtained from OpenStreetMap. The bus stops and intersections are projected onto the bus route, and a 20-meter buffer zone around each projected stop and intersection is defined as the area of influence \citep{jiang2024real}. Next, the coordinates and corresponding timestamps for each bus trip's trajectory are extracted from the GTFS Realtime data.

\begin{figure}[!ht]
  \centering
  \label{lab: figre1}
  \includegraphics[width=0.8\textwidth]{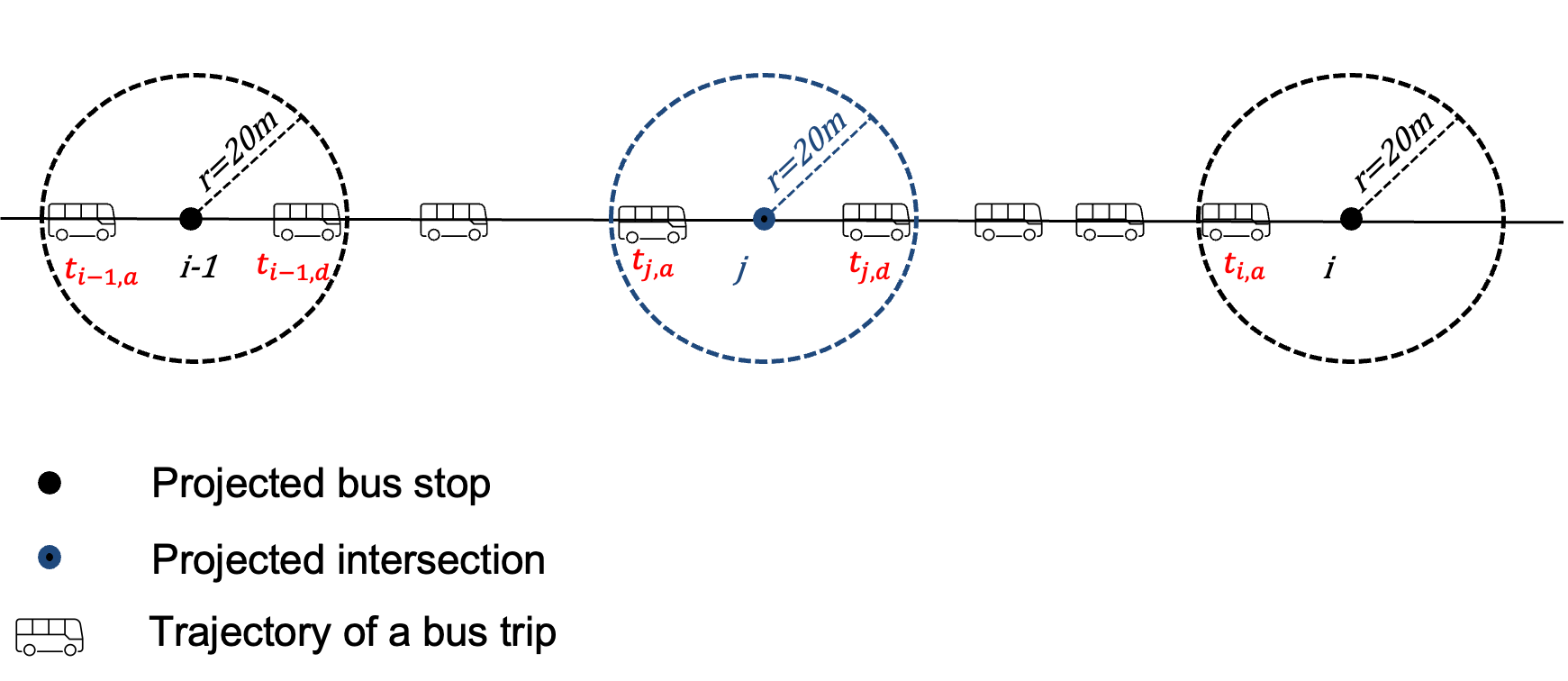}
  \caption{Example of bus stops, intersections, their buffer zones, and bus trajectories}\label{fig: link}
\end{figure}

As illustrated in Figure 1, for a given bus trip, when the bus first exits the 20-meter buffer zone of a projected stop $i$ (or a projected intersection $j$), the corresponding timestamp, $t_{i,d}$ ($t_{j,d}$), is recorded as the departure time from that stop (or intersection) \citep{liu2012bus}. Similarly, when the bus first enters the 20-meter buffer zone of stop $i$ (or intersection $j$), the corresponding timestamp, $t_{i,a}$ ($t_{j,a}$), is recorded as the arrival time at stop $i$ (or intersection $j$). The total travel time for the bus on the link between projected stops $i$-1 and $i$ is then estimated as:
\begin{equation}
    \Delta t_{li} = t_{i,d} - t_{i-1,d}.
\end{equation} 
The dwell time, $\Delta t_{di}$, at stop $i$ is inferred as:
\begin{equation}
    \Delta t_{di} = t_{i,d} - t_{i,a}.
\end{equation} 
If there are intersections between stops $i$-1 and $i$, the waiting and passing time, $\Delta t_{sj}$, at intersection $j$ is calculated as:
\begin{equation}
    \Delta t_{sj} = t_{j,d} - t_{j,a}.
\end{equation} 
Finally, if there are $J$ intersections between projected stops $i-1$ and $i$, the road travel time, $\Delta t_{ri}$, is estimated as:
\begin{equation}
    \Delta t_{ri} = t_{i,a}-t_{i-1,d}-\sum^J_{j=1} \Delta t_{sj}
\end{equation}

\subsection{Road Travel Time Prediction with Uncertainties}
In this section, we describe the approach for modeling the distribution of road travel time. Drawing on findings from existing literature \citep{he2018travel, hofmann2005impact, ma2019bus} and considering data availability, we incorporate four factors into the proposed model—peak hours, weekday/weekend variations, weather conditions (rain), and traffic congestion—as these are expected to influence the distribution of road travel time. The process of preparing the variables begins by extracting the time of day, day of the week, and weather conditions from the timestamps in the GTFS data and an external weather data source. 

\paragraph{\textbf{Time of day}} The time of day is recorded on an hourly basis and transformed into a binary variable to indicate whether the time falls within peak hours.
\paragraph{\textbf{Day of week}} The day of the week is similarly converted into a binary variable to distinguish weekdays from weekends.
\paragraph{\textbf{Weather condition}} Weather conditions, assumed to remain constant within each hour (as the resolution of weather data is hourly), are converted into a binary rain indicator to reflect the presence of rain or thunderstorms.
\paragraph{\textbf{Traffic congestion}} A traffic indicator is designed to capture instances when buses experience delays or stops due to traffic congestion. This factor is represented in real-time using GTFS data by calculating the space mean speed for each pair of records that fall within the bus link but outside the 20-meter buffer zone around bus stops and intersections (To ensure that the bus is traveling along the link, rather than dwelling at stops, waiting at intersections, or passing through intersections). The space mean speed is calculated as:
\begin{equation}
   v_k = \frac{d_{k,k-1}}{t_k-t_{k-1}},
\end{equation}
where $d_{k,k-1}$ represents the distance between the $(k)^{th}$ and $(k-1)^{th}$ records of a bus, and $t_k$ and $t_{k-1}$ are the corresponding timestamps. If the space mean speed on link $i$ for a given bus trip falls below a threshold value $V$, determined by spatial constraints and road characteristics, it is assumed that the bus is experiencing traffic congestion. In such cases, the traffic indicator for that link is set to 1. 

Once the factor preparation is complete, the modeling process is initiated to identify the distributions of the road travel times. Given the non-negative nature of road travel times and previous research indicating that road travel times generally follow a log-normal distribution \citep{uno2009using, duran2016estimation, kieu2015public}, this study adopts the log-normal distribution assumption for all links across the bus routes under various conditions. Subsequently, the two parameters of each log-normal distribution—mean and variance—are modeled as functions of the selected factors. The mean is expressed as a linear combination of these factors, while the variance is modeled by taking the exponential of a separate linear combination, following the approach outlined by \citet{rosopa2013managing}. This method is commonly used in transportation engineering to address heteroscedasticity \citep{greene2007heteroscedastic}. The detailed methodology is presented as follows:

Note that $\Delta t_{ri}$ denotes the road travel time on the link between stops $i-1$ and $i$, which is assumed to follow a log-normal distribution. By applying a logarithmic transformation, the variable $Y_{ri} = \text{ln}(\Delta t_{ri})$ is modeled as following a normal distribution with mean $\mu_{ri}$ and standard deviation $\sigma_{ri}$, and:
\begin{align}
    \mu_{ri} &= \bm{\beta}_{ri}^T\bm{X}_{ri}, \\
    \sigma_{ri}^2 &= \text{exp}(\bm{\gamma}_{ri}^T\bm{X}_{ri}).
\end{align}

Here, $\bm{\beta}_{ri} = [\beta_{ri0}, \beta_{ri1}, \dots, \beta_{ri4}]^T$, and $\bm{\gamma}_{ri} = [\gamma_{ri0}, \gamma_{ri1},\cdots, \gamma_{ri4}]^T$ are vectors of parameters to be estimated. $\bm{X}_{ri} = [X_{ri1}, X_{ri2}, X_{ri3}, X_{ri4}]^T$ represents the vector of corresponding variables for the link between stops $i-1$ and $i$. Specifically,
\begin{itemize}
    \item \(X_{ri1}\) is the rain indicator, with \(X_{ri1} = 1\) indicating rain or thunderstorm, and \(X_{ri1} = 0\) otherwise.
    \item \(X_{ri2}\) is the peak hour indicator, with \(X_{ri2} = 1\) indicating peak hours, and \(X_{ri2} = 0\) for non-peak hours.
    \item \(X_{ri3}\) is the weekday indicator, with \(X_{ri3} = 1\) indicating a weekday, and \(X_{ri3} = 0\) for weekends.
    \item \(X_{ri4}\) is the traffic indicator, with \(X_{ri4} = 1\) indicating the detection of slow movement or stopping of a bus on a link, and \(X_{ri4} = 0\) for the absence of slow movement or stopping.
\end{itemize}

Given that the log-transformed variable $Y_{ri} =$ ln$(\Delta t_{ri})$ follows a normal distribution influenced by various factors, the parameter set $\bm{\theta}_{ri} = [\bm{\beta}_{ri}, \bm{\gamma}_{ri}]^T$ is introduced, representing the coefficients associated with the mean and variance of $Y_{ri}$. We assume that observations of road travel time on a specific link are independent events. This implies that the running conditions for each bus on the road are randomly sampled throughout the day. This assumption aligns with the study’s objective to capture the general variability of road travel times over a broader time frame, rather than focusing on short-term dependencies. Additionally, this assumption is justified by the long headway of buses on the selected routes, which minimize temporal dependencies between consecutive buses. By leveraging the properties of the normal distribution, the log-likelihood function for the proposed model is formulated as:
\footnotesize
\begin{align}
 \ln(L(\mathbf{\bm{\theta_{ri}}})) &= -\frac{n}{2}\ln(2\pi) - \sum^n_{m=1}\ln(\sigma_{rim}) - \sum^n_{m=1}\frac{(y_{rim}-\mu_{rim})^2}{2\sigma^2_{rim}},
\end{align}
\normalsize
where $y_{rim}$ represents the $m_{th}$ observation of road travel time on the link between stops $i$-1 and $i$, while $\mu_{rim}$ and $\sigma_{rim}$ denote the actual mean and standard deviation corresponding to the $m_{th}$ observation, respectively. The estimates for $\bm{\theta}_{ri}$ are obtained by maximizing the log-likelihood function, $\ln{(L(\bm{\theta}_{ri}))}$, using numerical optimization techniques, as closed-form solutions are typically not available for this type of problem.

We proceed to prove the consistency of the MLE to justify its use in our road travel time prediction problem and to ensure, from a theoretical standpoint, that the model reliably produces accurate predictions based on the underlying data distribution.
\begin{proposition}
Under the specified assumptions (Assumption 1-2) and conditions (Condition 1-3) and Theorem 1, the MLE is consistent in the proposed model for estimating the parameters of the log-normal distribution of road travel times.
\end{proposition}

\begin{assumption}
The observations are independent and identically distributed (i.i.d.).
\end{assumption}

\begin{assumption}
The variables are linearly independent.
\end{assumption}

\begin{condition}
The observations are independent and identically distributed (i.i.d.).
\end{condition}

\begin{condition}
The log-likelihood function is continuously differentiable with respect to the parameter set $\bm{\theta}$.
\end{condition}

\begin{condition}
The Fisher Information matrix $\bm{F}(\bm{\theta})$ is positive definite.
\end{condition}

\begin{theorem}[Consistency of MLE \citep{rublik1995consistency}] Under the above regularity conditions (Condition 1-3), the MLE $\bm{\hat{\theta}_n}$ is consistent:
\begin{equation*}
\bm{\hat{\theta}_n} \overset{P}{\longrightarrow} \bm{\theta_0} \quad \text{as} \quad n \to \infty,
\end{equation*}
that is, if, for any $\epsilon$ $>$ 0, the probability that $\bm{\hat{\theta}_n}$ deviates from $\bm{\theta_0}$ by more than $\epsilon$ approaches zero as $n$ approaches infinity.
\end{theorem}

\begin{proof}
We establish the validity of Proposition 1 by demonstrating that the regularity conditions outlined in Theorem 1 are satisfied for the road travel time model.

\paragraph{\textbf{i.i.d}} Under Assumption 1, the observations are independent and identically distributed (i.i.d.).

\paragraph{\textbf{Continuity and differentiability}}
\footnotesize
\begin{align*}
\ln(L(\mathbf{\bm{\theta_{ri}}})) &= -\frac{n}{2}\ln(2\pi) - \frac{1}{2}\sum^n_{m=1}\ln(\sigma_{rim}^2) - \sum^n_{m=1}\frac{(y_{rim}-\mu_{rim})^2}{2\sigma^2_{rim}} \\
&= -\frac{n}{2}\ln(2\pi) - \frac{1}{2}\sum^n_{m=1}\bm{\gamma}_{ri}^T\bm{x}_{rim} \\
&\quad -\sum^n_{m=1}\frac{(y_{rim}-\bm{\beta}_{ri}^T\bm{x}_{rim})^2}{2\exp{(\bm{\gamma}_{ri}^T\bm{x}_{rim})}},
\end{align*}
\normalsize
where $\bm{x}_{rim} = [x_{ri1m}, x_{ri2m}, x_{ri3m}, x_{ri4m}]^T$ denotes the vector of the four independent variables for the $m_{th}$ observation. In the log-likelihood function, the first term is a constant. The second term is the summation of linear functions, while the third term is the summation of quotients. In each quotient, the numerator is the square of a linear function, and the denominator is the exponential of a linear function. Therefore, the log-likelihood function is a composition of continuous functions, and is thus continuous.
\begin{align*}
\frac{\partial \ln(L(\bm{\theta}_{ri}))}{\partial \bm{\beta}_{ri}} &= \bm{Z}_{ri}\bm{S}_{ri}\bm{R}_{ri} \\
\frac{\partial \ln(L(\bm{\theta}_{ri}))}{\partial \bm{\gamma}_{ri}} &= -\bm{Z}_{ri}\bm{1} + \bm{Z}_{ri}\bm{S}_{ri}\bm{Q}_{ri},
\end{align*}
where $\bm{Z}_{ri}$ represents a 4 $\times$ $n$ matrix containing the observations of the four independent variables. $\bm{S}_{ri} = \text{diag}(\frac{1}{\sigma_{ri1}^2},\frac{1}{\sigma_{ri2}^2},\dots,\frac{1}{\sigma_{rin}^2})$ is a $n$ $\times$ $n$ diagonal matrix with the inverse variances on the diagonal. $\bm{R}_{ri}$ is the vector of residuals, representing the difference between the observed values $y_{rim}$ and the corresponding actual means $\mu_{rim}$ across all observations. $\bm{Q}_{ri} = \bm{R}_{ri}^2$ represents the element-wise squared residuals. Therefore, by computing the partial derivatives with respect to the parameters, we establish that the log-likelihood function is differentiable with respect to the parameter vector $\bm{\theta_{ri}}$.

\paragraph{\textbf{Positive definiteness of the Fisher Information Matrix}}
Based on the properties of the Fisher Information matrix for a normal distribution, the Fisher Information matrix $F(\bm{\theta_{r,i}})$ can then be represented as:
\begin{align*}
\bm{F}(\bm{\theta}_{ri}) &= -E(\frac{\partial \ln{L(\bm{\theta}_{ri})}}{\partial \bm{\theta}_{ri}} \cdot \frac{\partial \ln{L(\bm{\theta}_{ri})}}{\partial \bm{\theta}_{ri}^T}) \\
&= \begin{bmatrix}
\bm{I}_{\beta\beta^T} & \bm{I}_{\beta\gamma^T} \\
\bm{I}_{\gamma\beta^T} & \bm{I}_{\gamma\gamma^T}
\end{bmatrix} = 
\begin{bmatrix}
\bm{I}_{\beta\beta^T} & \bm{0} \\
\bm{0} & \bm{I}_{\gamma\gamma^T}
\end{bmatrix} \\
\bm{I}_{\beta\beta^T} &= \bm{Z}_{ri}^T \bm{S}_{ri} \bm{Z}_{ri} \\
\bm{I}_{\gamma\gamma^T} &= \bm{Z}_{ri}^T \bm{C}_{ri} \bm{Z}_{ri},
\end{align*}
where $\bm{C}_{ri}$ is a $n$ $\times$ $n$ diagonal matrix with each diagonal entry equal to $\frac{1}{2}$. Under Assumption 2, where the variables are linearly independent, $\bm{x}_{ri}$ has full rank. As a result, the Fisher Information submatrices $\bm{I}_{\beta\beta^T}$ and $\bm{I}_{\gamma\gamma^T}$ are positive definitie. Consequently, the entire Fisher Information matrix is positive definite. By verifying that the regularity conditions are satisfied and applying Theorem 1, the consistency of the MLE for the road travel time model is established. 
\end{proof}

For a new observation with given values of the rain indicator ($X_{ri1}$), peak hour indicator ($X_{ri2}$), weekday indicator ($X_{ri3}$), and traffic indicator ($X_{ri4}$), the road travel time is predicted by evaluating the predicted mean of $Y_{ri}$. Let the MLE of the parameters be denoted by $\hat{\bm{\beta}}_{ri}$ and $\hat{\bm{\gamma}}_{ri}$. The predicted road travel time $\hat{\Delta t}_{ri}$ is the predicted median of the log-normal distribution (the exponential of the predicted mean of the normal distribution):
\begin{equation}
    \hat{\Delta t}_{ri} = \text{exp}(\hat{\mu}_{ri}) = \text{exp}(\bm{\beta}_{ri}^T\bm{x}_{ri}).
\end{equation}
Where $\bm{x}_{ri}$ = [$x_{ri1}, x_{ri2}, x_{ri3}, x_{ri4}]^T$ represents the vector corresponding to a single observation of the four independent variables.

In this study, the traffic indicator is a variable that can be updated in real-time to reflect the bus's real-time traffic conditions on the road. Specifically, as GTFS Realtime data is reported approximately every 15-20 seconds, the space mean speed, $v_k$, can be subsequently updated. When the speed $v_k$ on the link between stops $i-1$ and $i$ falls below the predefined threshold $V$, the traffic indicator switches from zero to one, and conversely, it switches from one to zero when the speed rises above $V$. The model generates new predictions in real time whenever the traffic indicator is updated.

In addition to providing point estimates for road travel times, the model is also capable of generating uncertainty bounds (confidence intervals) for these estimates. The confidence intervals are based on the asymptotic normality of the estimator, with a covariance matrix derived from the Fisher Information Matrix for $\bm{\theta}_{ri}$. As outlined in the proof of Proposition 1, the Fisher Information is defined as:
\begin{equation}
F(\bm{\theta}_{ri}) = -E(\frac{\partial \ln{L(\bm{\theta}_{ri})}}{\partial \bm{\theta}_{ri}} \cdot \frac{\partial \ln{L(\bm{\theta}_{ri})}}{\partial \bm{\theta}_{ri}^T})
\end{equation}
Let $\bm{A}=\begin{bmatrix} \bm{I_5} & \bm{0} & \bm{0} & \bm{0} & \bm{0} & \bm{0} \end{bmatrix}$, where $\bm{0}$ is a 5 $\times$ 1 vector, and $\bm{A}$ is a 5 $\times$ 10 matrix. Then, $\hat{\mu}_{ri}$ asymptotically follows the distribution:
\begin{equation}
    \hat{\mu}_{ri} \xrightarrow{\text{d}} N(\mu_{ri}, \frac{[1, \bm{x}_{ri}^T]\bm{A}\bm{F}^{-1}(\bm{\theta_{r,i}})\bm{A^T}[1,\bm{x}_{ri}^T]^T}{n_r}),
\end{equation}
where $n_r$ is the number of observations in the road travel time dataset. It is important to emphasize that the uncertainty bounds for $\mu_{ri}$ depend on the specific combination of covariates. The 95\% confidence interval, which serves as the uncertainty bound for the link travel time under different combinations of covariates, can be estimated. The 95\% confidence interval is defined as:
\begin{align}
    [\exp{(\mu_{ri}-1.96 \sigma(\hat{\mu}_{ri}))}, \exp{(\mu_{ri}+1.96 \sigma(\hat{\mu}_{ri}))}].
\end{align}
Note that $\sigma(\hat{\mu}_{ri}))$ represents the standard deviation of the distribution defined in Equation (12).

\subsection{Stop Dwell Time and Intersection Waiting and Passing Time Prediction}
It is important to highlight that drivers may skip stops in cases of bus delays or when no passengers are boarding or alighting. As a result, the distribution of stop dwell times may exhibit one peak at zero and another peak corresponding to typical dwell times when stops are not skipped. This bimodal behavior makes it unsuitable to fit the stop dwell time using a single unimodal distribution. Consequently, stop dwell times are predicted directly based on their empirical distribution. In contrast, intersection waiting and passing times are assumed to follow a log-normal distribution, which is characterized by its asymmetry and right-skewness. Log-normal distributions are fitted to the intersection waiting and passing times using the parameters $\mu_{si}$ and $\sigma_{si}$. The predicted intersection waiting and passing times are then obtained by taking the exponential of $\hat{\mu}_{si}$. 

\subsection{Markovian Framework}
Consider a transit network comprising $p$ links, at any discrete time point $q$, the state of the bus running a link can be encapsulated by the state vector
\begin{equation}
\bm{Z}_{q} = [z_{q1},\dots, z_{qp}].
\end{equation}
Each component $z_{qi}$ within the state vector $\bm{Z_{q}}$ is binary, indicating the bus running on link $i$ \citep{zhao2017modeling}. Specifically, when the bus is traversing link $i$ at time point $q$, $z_{qi}$ is set to 1, with all other components set to 0. Thus, the state vector $\bm{Z_{q}}$ can be represented as $\bm{Z_{q}}$ = $[0,\dots,1,\dots,0]$, where the 1 occurs at the $i_{th}$ position. The sequence of the Markov process is
\begin{equation}
P(\bm{Z_{q+1}}|\bm{Z_{0}},\bm{Z_{1}},\dots,\bm{Z_{q}})=P(\bm{Z_{q+1}}|\bm{Z_{q}}).
\end{equation}
Therefore, the probability of moving to the next state at time step $q$+1 depends only on the present state at time step $q$ and not on the previous state.

A bus route comprising $p$ links can be fully characterized by a Markov process, which is effectively captured through a transition matrix $\bm{P}$ of dimensions $p$-by-$p$, where transition probabilities are derived from the predicted road travel times. This Markovian framework, combined with the distributions of stop dwell time and intersection waiting/passing times, enables the simulation of bus movement along its route. A key contribution of this study is the development of an approach that enables real-time prediction of both the remaining travel time to each future stop and the associated uncertainty range at each stop. The process for defining the transition probabilities within the Markov framework, as well as the estimation methods for stop dwell, and intersection waiting and passing times, are detailed as follows.

As depicted in Figure 2, at a given time step $\tau$, a bus situated on link $BE$ has the possibility to either remain within link $BE$ or proceed to stop $E$ and then transition to the adjacent link $EH$ at time $\tau$ + 1. Given that the bus is constrained from diverting to alternative links, in the $i_{th}$ row of the transition probability matrix, only two probabilities, $(p_{i,i})_{\tau}$ (the probability of staying on the current link) and $(p_{i,i+1})_{\tau}$ (the probability of moving to the next link), hold non-zero values, while all other probabilities in this row are set to zero. Hence, $(p_{i,i})_{\tau}$ + $(p_{i,i+1})_{\tau}$=1. Once $(p_{i,i})_{\tau}$ is determined, $(p_{i,i+1})_{\tau}$ is automatically set to one minus $(p_{i,i})_{\tau}$. Therefore, the primary focus is on estimating $(p_{i,i})_{\tau}$. 

\begin{figure}[!ht]
  \centering
  \includegraphics[width=0.6\textwidth]{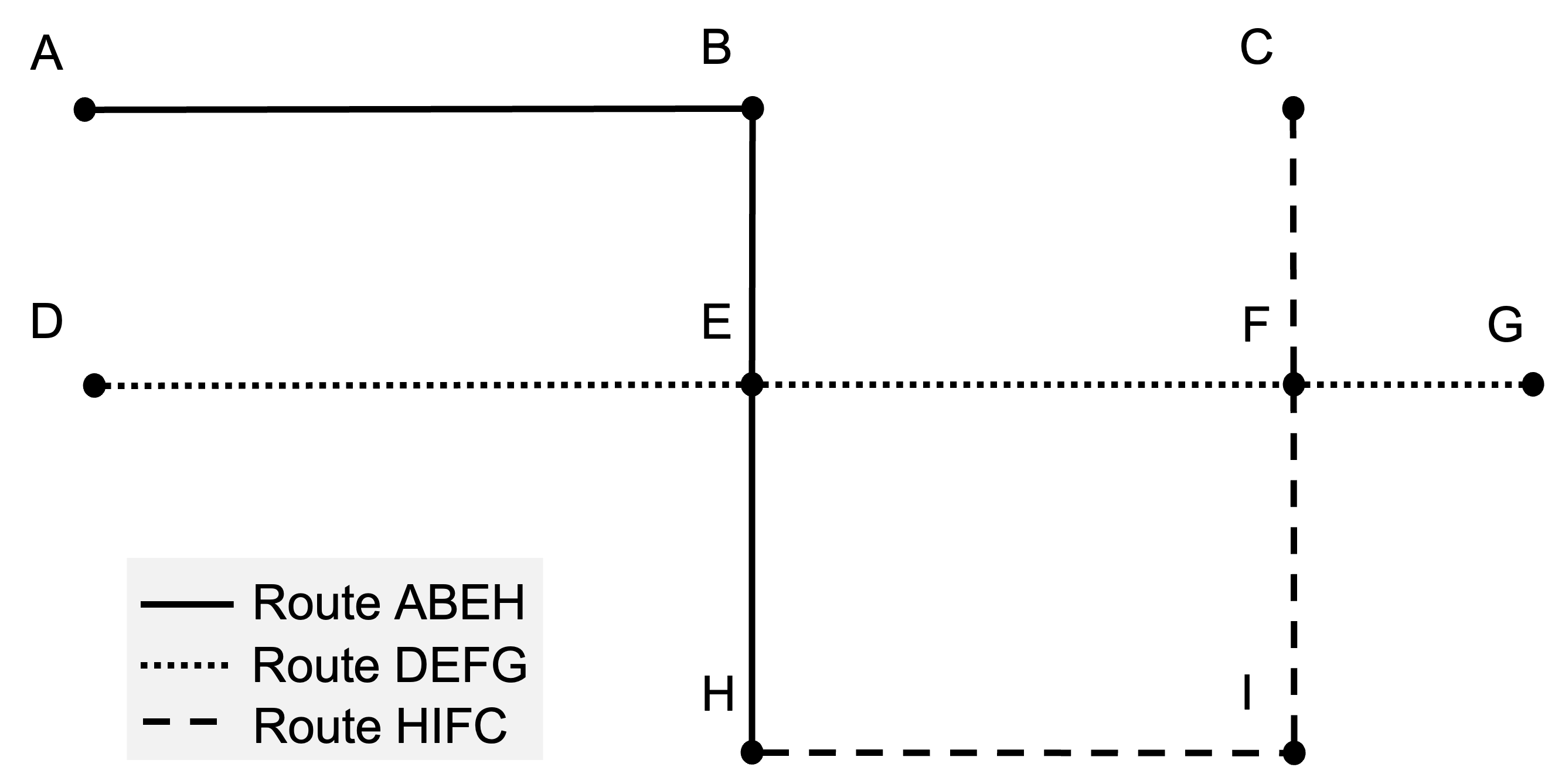}
  \caption{Transit Network}\label{fig: transit network}
\end{figure}

We first apply the method proposed in the previous step to predict road travel times using the training dataset. Given the known link lengths, these predictions are then used to estimate bus speed $v_{ri}$ on link $i$ and subsequent links while ensuring that all covariates remain consistent with those extracted at time $\tau$. Given this, the number of time steps required to complete the link, denoted by $S_i$, is calculated as $S_i$ = $\frac{d_i-d_{\tau}}{\Delta Tv_{ri}}$. Similarly, for subsequent links, $S_i$ is determined by dividing the link length by the predicted travel distance within each $\Delta T$ interval. The corresponding estimate for the probability $\hat{p}_{i,i}$ of the bus remaining on the current link $i$ is then derived as follows:
\begin{equation}
\hat{p}_{i,i} = \frac{S_i-1}{S_i}
\end{equation}

Then road travel time on link $i$ is calculated as $f_{i}\Delta T$ where $f_i$ represents the number of time units the bus remains on link $i$, directly obtained from the Markov simulation results. To simulate the stop dwell time or intersection waiting and passing time, we generate these times randomly based on their fitted distributions. The simulated dwell time $\Delta t_{di}$ at stop $i$ and waiting and passing time time $\Delta t_{sj}$ at intersection $j$ are then summed with the predicted road travel time to calculate the total remaining time $\Delta t_{li}$ to reach stop $i$. To ensure robustness, this simulation is repeated $M$ times, with the mean value of $\Delta t_{li}$ across all simulations used as the predicted total remaining time. 

We proceed by establishing the identifiability of the transition matrix $\bm{P}$ under the following conditions:

\begin{proposition}
Given that Conditions 5 is met, the transition matrix $\bm{P}$ is identifiable.
\end{proposition}
\begin{condition}
The parameters $\bm{\theta_{ri}}$ for each link $i$ are identifiable.
\end{condition}

\begin{proof}{Proof}
Condition 5 is guaranteed by Proposition 1. According to this condition, the parameters $\bm{\theta_{ri}}$ for each link $i$ are identifiable. Then the road travel time for each link can be estimated via Equation (10). By selecting an appropriate time unit $\Delta T$, such that $\Delta T$ is less than the estimated travel time across any specific link on the bus route, the diagonal entries of the transition matrix $\bm{P}$ for the links along the route are identifiable, according to Equation (16). Since each off-diagonal entry $p_{i,i+1}$ can be directly derived from the corresponding diagonal entry $p_{i,i}$, the entire transition matrix $\bm{P}$ is identifiable.
\end{proof}

The expected time interval for the bus to reach stop $i$ is determined using the 2.5th and 97.5th percentiles of the $M$ Markov simulation results.

\section{Empirical Analysis}
\subsection{Study Areas}
For the empirical analysis, three bus routes in Gainesville, Florida, were selected. These routes were then used to evaluate the proposed model’s performance in predicting travel times and quantifying uncertainties. Gainesville was chosen due to its well-established public transportation system, which includes 38 strategically planned bus routes covering the city’s key network links in both directions. The first two routes represent the two directions of Route 1: one heading towards the downtown station and the other towards Butler Plaza, a major commercial hub. The third route is Route 75, which leads to the Oaks Mall, located away from the center of Gainesville and the downtown area. These routes operate at relatively low frequencies, with Route 1 running approximately one bus every 30 minutes and Route 75 running approximately one bus every hour. The selected routes are illustrated in Figure 3.
\begin{figure}[!ht]
  \centering
  \includegraphics[width=0.7\textwidth]{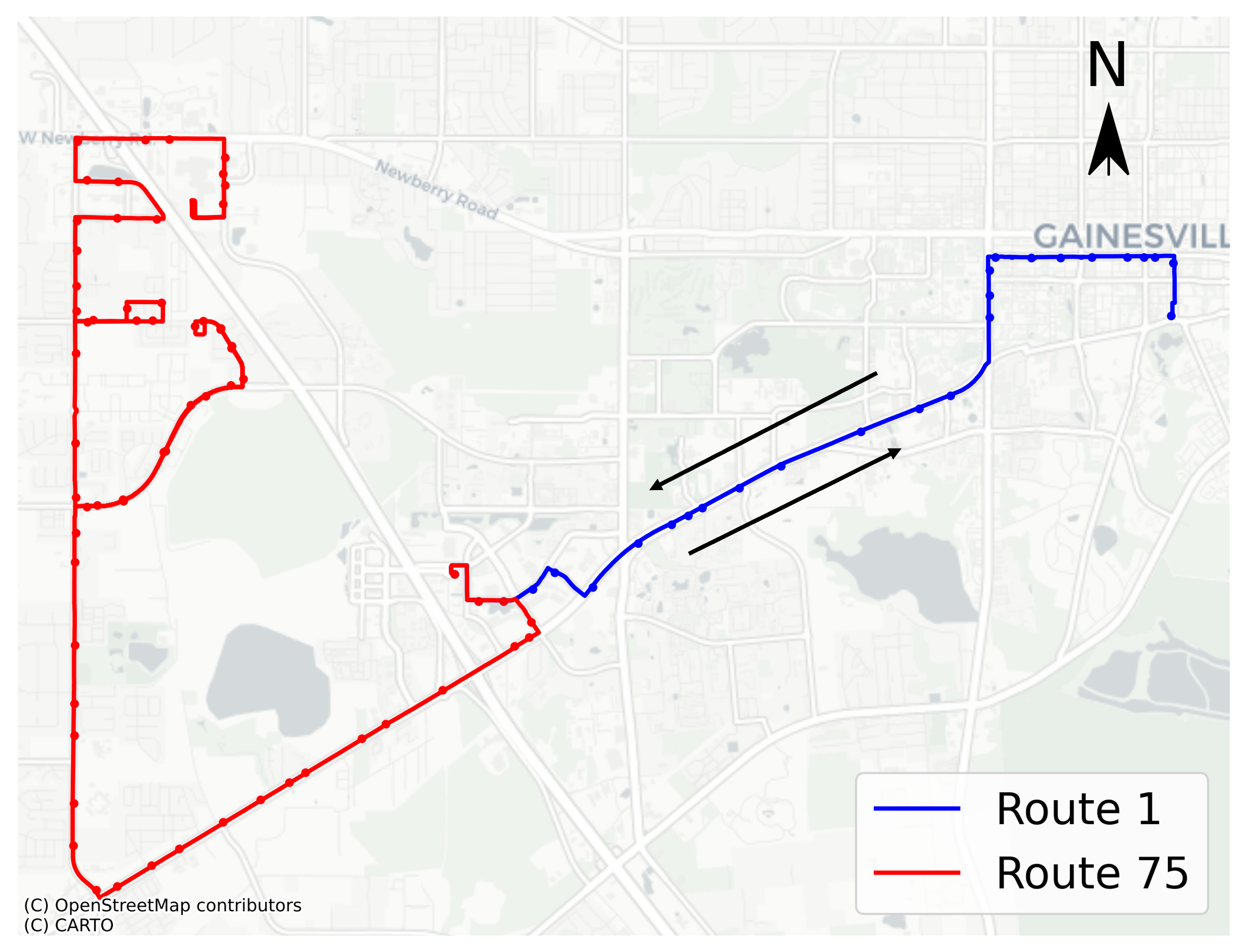}
  \caption{Selected Bus Routes}\label{fig: route graph}
\end{figure}

\subsection{Data Acquisition}
The empirical analysis and proposed methodology were applied and validated using GTFS data, along with historical weather data collected from August 18, 2023, to October 15, 2023, in Gainesville, Florida. Specifically, the GTFS data comprises both static and real-time components. The static data, which includes information about bus routes and stops, was obtained from the Gainesville Department of Transportation \citep{GainesvilleStaticGTFSData}, while the GTFS Realtime data was sourced from the BusTime Developer API \citep{bustime2022}. The real-time records, updated every 15-20 seconds, cover observation hours ranging from 6 a.m. to 11 p.m. Additionally, historical weather data was downloaded from OpenWeather (https://openweathermap.org/). The complete dataset was used for the empirical analysis. For the model performance evaluation, the data were subsequently divided into two sets: data from August 18th to October 8th were used to train the proposed model, and data from the last week (October 9th to October 15th) were used for validation purposes. In the study, the location and length of each bus link were derived from GTFS static data. Link travel times, as well as factors such as peak hour, weekend, and traffic indicators, were derived from GTFS Realtime data. In addition, weather indicators were extracted from historical weather data.

\subsection{Statistical Analysis of Empirical Distributions}
Statistical tests are performed on the empirical distributions of road travel time and intersection waiting and passing time to assess the validity of the assumptions underlying the proposed methodology. The test results for the time components of five randomly selected links, each with varying characteristics such as link length, road type, and the presence of intersections, are presented. Two intersections are present among these five links. The key characteristics of the five links are as follows:
\begin{itemize}
    \item Link1: Route 1, Begin at the starting station, follow local roads, and have no signal lights along the link.
    \item Link2: Route 1, Located on local roads, with a signal light present on the link.
    \item Link3: Route 1, Asteroids, with a signal light present on the link.
    \item Link4: Route 1, Situated on local roads in the opposite direction, a very short link with no signal lights present.
    \item Link5: Route 75, Located on state road, with no signal light present on the link.
\end{itemize}

\subsubsection{Evaluation of Log-Normal Distribution Assumption}
Road travel times and intersection waiting and passing times are assumed to follow a log-normal distribution, characterized by a long tail and right-skewness. To validate this assumption, the Kolmogorov-Smirnov (K-S) test is applied. The null hypothesis for the K-S test states that the respective time component follows a log-normal distribution. The test results are displayed in Table 1.

\begin{table}[!ht]
\centering
\begin{tabular}{lll}
\hline
                                            & K-S Statistics & p-value \\ \hline
\textbf{Road Travel Time}                   &                &         \\
Link 1                                      & 0.035          & 0.559   \\
Link 2                                      & 0.046          & 0.235   \\
Link 3                                      & 0.039          & 0.418   \\
Link 4                                      & 0.054          & 0.108   \\
Link 5                                      & 0.027          & 0.833   \\ \hline
\textbf{Intersection Waiting and Passing Time} &                &         \\
Intersection 1                              & 0.039          & 0.433   \\
Intersection 2                              & 0.044          & 0.283   \\ \hline
\end{tabular}
\caption{K-S Test Results}
\end{table}

The $p$-values are all greater than 0.05, indicating that the null hypothesis cannot be rejected at the 0.05 significance level for the time components of all selected links and intersections. This supports the assumption that road travel times and intersection waiting and passing times follow a log-normal distribution.

\subsubsection{Evaluation of Heteroscedasticity Assumption}
Heteroscedasticity is assumed in the proposed log-normal distribution model. To assess the validity of this assumption, the Breusch-Pagan (BP) test is conducted on the log-transformed road travel times for the five selected links. The BP test examines whether the variance of the residuals is dependent on one or more independent variables. The null hypothesis posits that the residuals exhibit constant variance (homoscedasticity). The results of the BP test for the selected links are summarized in Table 2.

\begin{table}[H]
\centering
\begin{tabular}{lll}
\hline
       & LM Statistics & p-value          \\ \hline
Link 1 & 22.040        & \textless{}0.001 \\
Link 2 & 30.093        & \textless{}0.001 \\
Link 3 & 23.022        & \textless{}0.001 \\
Link 4 & 10.241        & 0.037            \\
Link 5 & 21.288        & \textless{}0.001 \\ \hline
\end{tabular}
\caption{Breusch-Pagan Test Results}
\end{table}

All p-values are below the 0.05 significance threshold, with four links exhibiting p-values less than 0.001. These results indicate that the null hypothesis of constant variance (homoscedasticity) is rejected for all five links, providing strong evidence of the presence of heteroscedasticity.

\subsubsection{Evaluation of Independence Assumption}
The road travel times are assumed to be independent over a broader temporal framework. To evaluate this assumption, the Runs Test is employed, which assesses the randomness or independence of a sequence of observations. The null hypothesis for the Runs Test posits that the sequence is random. For the application of the test, road travel times on on any link $i$ between adjacent stops $i$-1 and $i$ are sorted by the departure time from stop $i$-1. The results of the Runs Test are presented in Table 3.

\begin{table}[H]
\centering
\begin{tabular}{lll}
\hline
       & Z-Statistics & p-value \\ \hline
Link 1 & -1.553       & 0.120   \\
Link 2 & -0.689       & 0.491   \\
Link 3 & -1.001       & 0.317   \\
Link 4 & -1.832       & 0.067   \\
Link 5 & -1.726       & 0.084   \\ \hline
\end{tabular}
\caption{Runs Test Results}
\end{table}

The results show that the p-values are greater than 0.05 for all links, indicating that the assumption of randomness in the sequence cannot be rejected at the 0.05 significance level. This supports the validity of the independence assumption for the road travel times. However, for two links, the null hypothesis cannot be rejected at the 0.1 significance level, suggesting that short-term temporal correlations or trends may marginally exist. Nonetheless, these results are overall reasonable given the broader time framework under consideration, where minor deviations from strict randomness do not significantly undermine the assumption of independence.

\section{Results and Discussion}
\subsection{Road Travel Time Prediction with Uncertainties}
Tables 4 and 5 present the comparative results of road travel time predictions and the corresponding uncertainty bounds across various methods, including the proposed approach and several baseline models. A detailed description of the baseline models and performance evaluation metrics can be found in Appendix A and Appendix B, respectively. In these tables, LN-MLE refers to the proposed method, which predicts road travel times by modeling a log-normal distribution and estimating its parameters using the MLE technique, HM stands for the Historical Mean model, LR represents the Linear Regression model, DT refers to Decision Tree, and MLP indicates the Multi-Layer Perceptron model. Additionally, the column labeled BW in Table 5 denotes the width of the uncertainty bounds. It is important to highlight that the 95\% confidence intervals for all methods are calculated based on the most frequently occurring combinations of covariates. 

\begin{table}[H]
\footnotesize
\centering
\begin{tabular}{lllllllllll}
\hline
       & \multicolumn{2}{c}{LN-MLE}        & \multicolumn{2}{c}{HM} & \multicolumn{2}{c}{LR} & \multicolumn{2}{c}{DT} & \multicolumn{2}{c}{MLP} \\ \hline
       & MAE             & RMSE            & MAE        & RMSE      & MAE        & RMSE      & MAE        & RMSE      & MAE        & RMSE       \\ \hline
Link 1 & \textbf{13.068} & \textbf{17.203} & 23.851     & 28.122    & 14.338     & 18.368    & 14.372     & 18.629    & 14.306     & 18.342     \\
Link 2 & \textbf{5.935}  & \textbf{11.293} & 10.853     & 17.075    & 6.782      & 12.125    & 6.612      & 11.913    & 6.763      & 12.089     \\
Link 3 & \textbf{42.555} & \textbf{55.639} & 46.708     & 59.725    & 43.971     & 56.626    & 43.687     & 56.371    & 44.035     & 56.536     \\
Link 4 & \textbf{2.821}  & \textbf{3.604}  & 3.982      & 4.609     & 3.021      & 3.985     & 3.019      & 3.961     & 3.062      & 3.987      \\
Link 5 & \textbf{11.040} & \textbf{18.513} & 27.326     & 34.220    & 11.837     & 19.140    & 12.044     & 19.100    & 11.955     & 19.163     \\ \hline
\end{tabular}
\caption{Comparison of Prediction Results between the Proposed Model and Baseline Models}
\end{table}

\begin{table}[H]
\small
\centering
\begin{tabular}{lllllll}
\hline
       & \multicolumn{2}{c}{LN-MLE}            & \multicolumn{2}{c}{HM}          & \multicolumn{2}{c}{LR}          \\ \hline
       & 95\%CI               & BW             & 95\%CI                 & BW     & 95\%CI                 & BW     \\ \hline
Link 1 & {[}25.227, 27.318{]} & \textbf{2.091} & {[}26.628, 29.324{]}   & 2.696  & {[}26.617, 29.336{]}   & 2.719  \\
Link 2 & {[}40.801, 47,459{]} & \textbf{6.757} & {[}39.658, 53.260{]}   & 13.601 & {[}33.872, 47.881{]}   & 14.009 \\
Link 3 & {[}75.512, 80.877{]} & \textbf{5.364} & {[}152.913, 165.507{]} & 12.594 & {[}152.624, 164.761{]} & 12.138 \\
Link 4 & {[}8.993, 9.662{]}   & \textbf{0.669} & {[}9.800, 10.725{]}    & 0.925  & {[}9.445, 10.835{]}    & 1.389  \\
Link 5 & {[}26.397, 29.347{]} & \textbf{2.950} & {[}24.379, 31.424{]}   & 7.005  & {[}25.693, 31.072{]}   & 5.379  \\ \hline
\end{tabular}
\caption{Comparison of Uncertainty Bounds between the Proposed Model and Baseline Models}
\end{table}

The results indicate that, although the proposed model does not significantly outperform the baseline models, it consistently achieves superior performance across all links. In contrast, the baseline methods exhibit varying performance rankings depending on the specific links. This consistent and satisfactory performance is particularly valuable for predicting road travel times, which are inherently stochastic in nature. The finding highlight the proposed model's ability to reliably and effectively capture the underlying patterns and dynamics of road travel times.

The comparison of the 95\% confidence intervals and their corresponding widths between the proposed model and the baseline models reveals that the proposed model consistently generates narrower uncertainty bounds. This indicates that the proposed model provides predictions with a higher level of confidence in their accuracy compared to the baseline methods. Moreover, the results highlight the model's improved capability to effectively capture and quantify the inherent variability in the data.

\subsection{Factor Contributions to Road Travel Time Prediction}
Table 6 presents the coefficients of each factor in the road travel time distribution model.  In the table, *** denotes a $p$-value $< 0.001$, while * denotes a $p$-value $< 0.05$, indicating statistical significance.

\begin{table}[H]
\centering
\begin{tabular}{lllllll}
\hline
 &  & Link 1 & Link 2 & Link 3 & Link 4 & Link 5 \\ \hline
\multicolumn{1}{c}{\multirow{4}{*}{$\mu_{ri}$}} & Const. & \textbf{3.314***} & \textbf{3.716***} & \textbf{4.278***} & \textbf{2.257***} & \textbf{3.335***} \\
\multicolumn{1}{c}{} & Rain & 0.027 & \textbf{0.027*} & -0.003 & -0.006 & -0.003 \\
\multicolumn{1}{c}{} & Peak Hour & \textbf{0.096***} & \textbf{0.044***} & \textbf{0.062***} & \textbf{0.067*} & \textbf{0.092***} \\
\multicolumn{1}{c}{} & Weekday & 0.046 & \textbf{0.067*} & 0.081 & -0.025 & -0.009 \\
\multicolumn{1}{c}{} & Traffic & \textbf{0.961***} & \textbf{0.613***} & \textbf{0.732***} & \textbf{1.923***} & \textbf{1.106***} \\ \hline
\multirow{4}{*}{$\sigma_{ri}^2$} & Const. & \textbf{-2.461***} & \textbf{-3.344***} & \textbf{-3.302***} & \textbf{-2.104***} & \textbf{-2.470***} \\
 & Rain & -0.015 & 0.148 & -0.106 & -0.146 & \textbf{0.272*} \\
 & Peak Hour & -0.110 & -0.030 & \textbf{0.173*} & 0.123 & \textbf{0.269*} \\
 & Weekday & -0.028 & \textbf{-0.502*} & 0.135 & -0.221 & \textbf{-0.703*} \\
 & Traffic & \textbf{0.404***} & \textbf{1.067***} & \textbf{1.118***} & \textbf{0.314***} & \textbf{0.250*} \\ \hline
\end{tabular}
\caption{Factor Contributions to the Road Travel Time Distribution for Each Link}
\label{tab:factor_contributions}
\end{table}

Table 6 demonstrates that the factors impacting the parameters of road travel time distributions vary across links with different characteristics. Among the factors, traffic congestion and peak hours emerge as the dominant contributors, with traffic congestion significantly affecting all parameters across all links. For asteroids or state roads, the traffic condition may be more fluctuated during peak hours, making it a significant impacting factor on the variability of road travel time distributions. Conversely, it is always not significant for the variability of road travel times on local roads. On arterial and state roads, traffic conditions tend to fluctuate more during peak hours, making peak hour a key determinant of variability in road travel time distributions. In contrast, on local roads, traffic conditions generally do not significantly impact variability. For some local roads, rain increases road travel times. On some local roads, rain contributes to increased road travel times, potentially due to narrower lanes and the tendency of bus drivers to adopt more risk-averse behavior by reducing speed for safety. Rain can also increase the variability of road travel times, particularly when buses are traveling on state roads. This may be attributed to the higher speed limits on state roads, which cause rain-induced slowdowns to vary significantly among different buses. Weekdays are associated with increased road travel times in urban areas, particularly on local roads. However, they may also reduce road travel time variability on certain links. This is likely due to the more structured traffic flow on weekdays compared to weekends, when irregular travel patterns—such as shopping and recreational activities—introduce greater variability in traffic conditions.

\subsection{Markovian Framework}
To illustrate the effectiveness of the real-time remaining time predictions and the corresponding time range, Figure 4 presents the comparison between the actual remaining time, the time predicted by the historical mean method, and the remaining time along with its potential range predicted by the Markovian framework at four consecutive timestamps. 

\begin{figure}[!ht]
    \centering
    % First row
    \begin{subfigure}{0.4\textwidth}
        \centering
        \includegraphics[width=\linewidth]{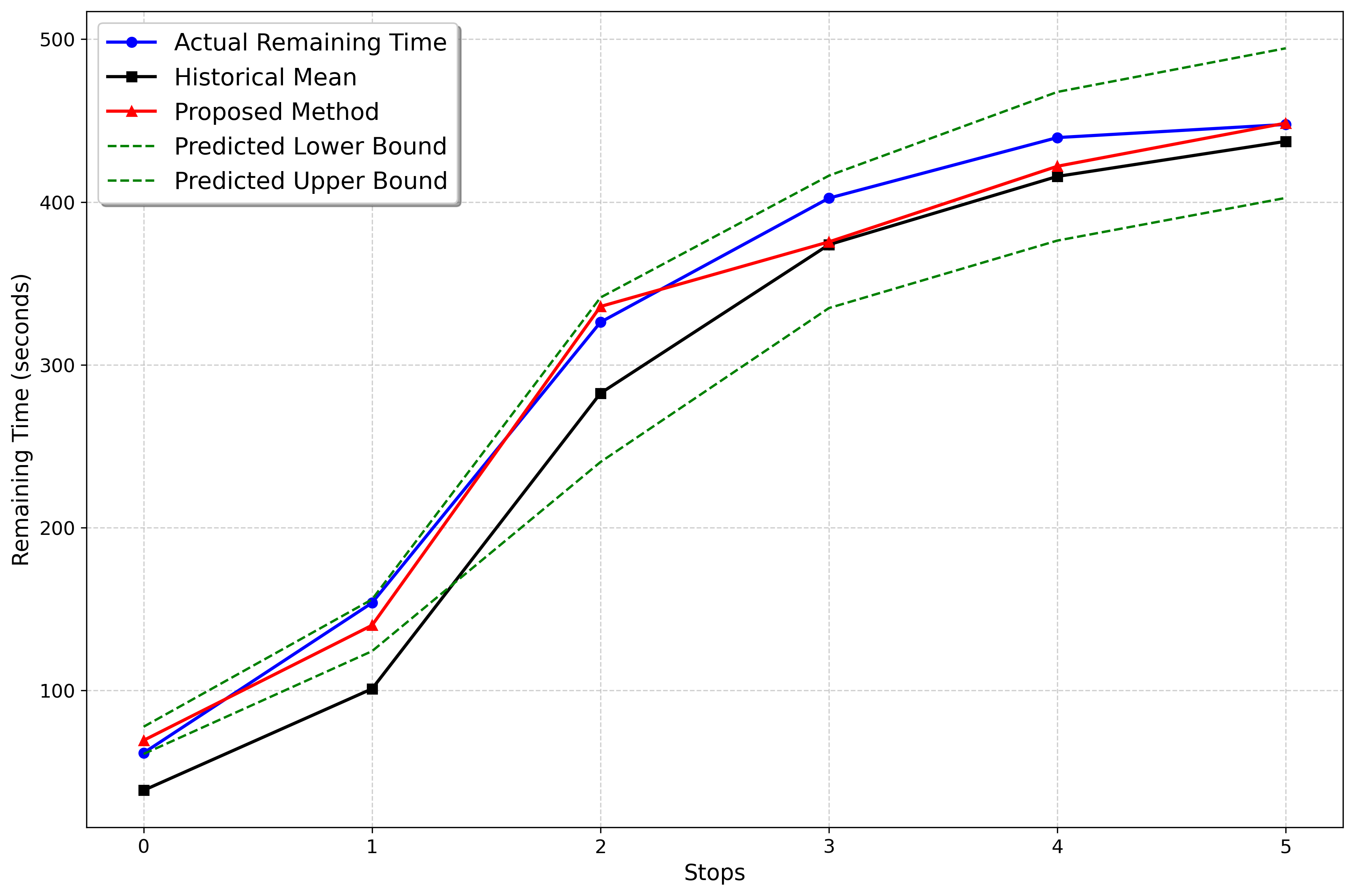}
        \caption{Timestamp 1}
        \label{fig:sub1}
    \end{subfigure}
    \hfill
    \begin{subfigure}{0.4\textwidth}
        \centering
        \includegraphics[width=\linewidth]{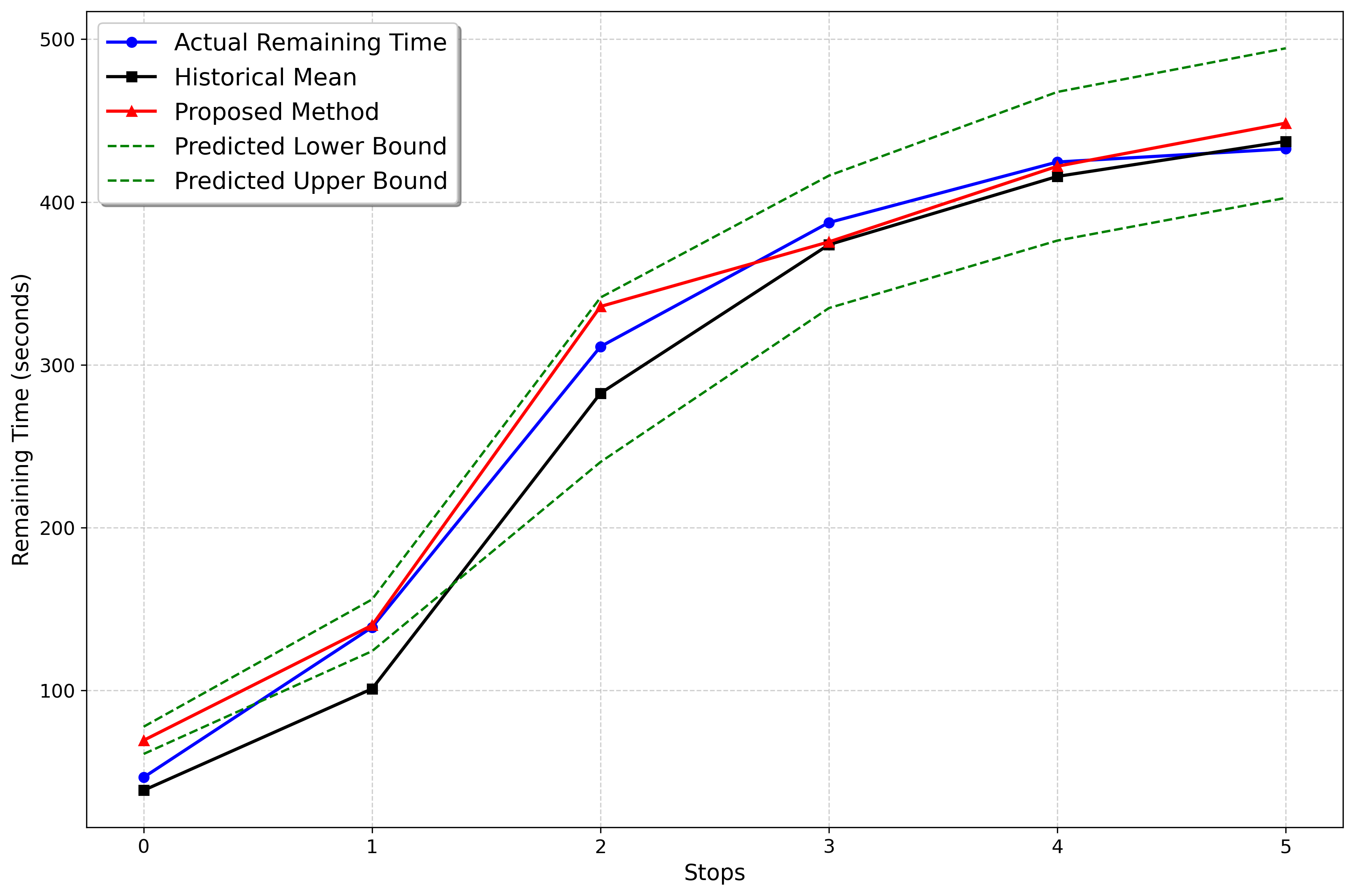}
        \caption{Timestamp 2}
        \label{fig:sub2}
    \end{subfigure}
    
    % Second row
    \begin{subfigure}{0.4\textwidth}
        \centering
        \includegraphics[width=\linewidth]{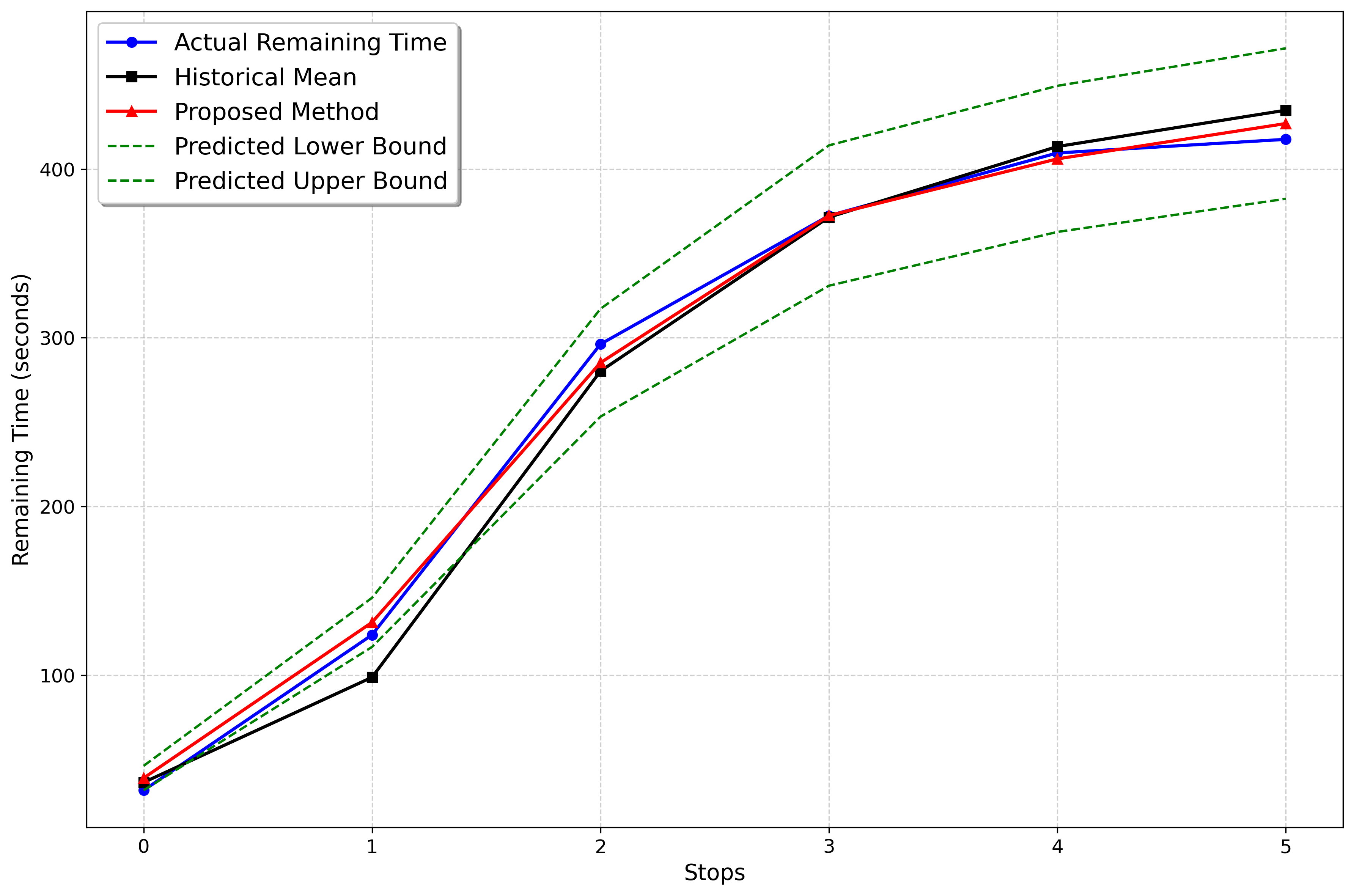}
        \caption{Timestamp 3}
        \label{fig:sub3}
    \end{subfigure}
    \hfill
    \begin{subfigure}{0.4\textwidth}
        \centering
        \includegraphics[width=\linewidth]{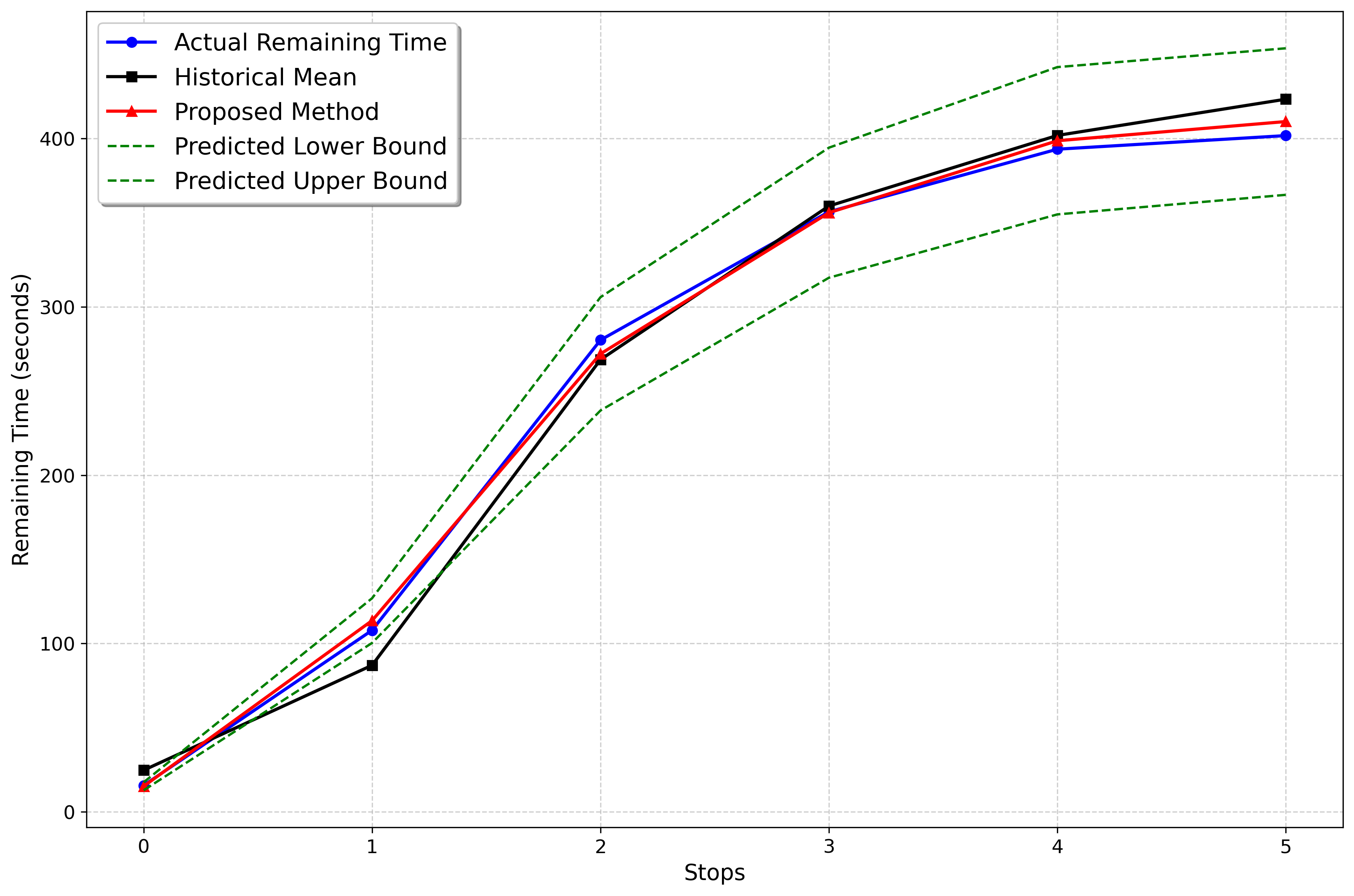}
        \caption{Timestamp 4}
        \label{fig:sub4}
    \end{subfigure}
    \caption{Real-time Remaining Time Prediction Comparison}
    \label{fig:main}
\end{figure}

The results demonstrate that the Markovian framework closely approximates the real remaining time and significantly outperforms the historical mean method by accounting for dynamic traffic conditions. While certain predictions might be affected by the assumption that all links share the same covariates at a given timestamp, the potential time intervals generated by the Markov model reliably encompass the actual remaining time. This indicates the robustness of the system, offering passengers more accurate real-time updates to improve trip planning.

\section{Conclusion}
This study introduces a novel hybrid methodology that integrates the MLE process with a stationary Markov Chain to predict bus link travel times and travel times from any current location to any future stop in real-time, while also quantifying their associated uncertainties. A key contribution is the identification and incorporation of a traffic indicator that signals whether buses are delayed on a link, utilizing only publicly available GTFS data. In the analysis, each link's travel time is divided into three components: road travel time, stop dwell time, and intersection waiting and passing time, with each component modeled independently. It is assumed that road travel time and intersection waiting and passing time follow log-normal distributions, and the parameters for road travel time are modeled as functions of various factors, accounting for heteroscedasticity. MLE, along with FIM, is employed to estimate these parameters, predict road travel times, and quantify uncertainties. The parameters have been demonstrated to be consistent, indicating the reliability of the predictions. Additionally, a stationary Markov model is proposed to predict travel times between any en-route location to any future stop while assess their uncertainties. The transition probabilities are calculated using the MLE results, and the identifiability of the transition matrix further enhances the reliability of the findings. The model's performance is evaluated by comparing it against baseline models using the real-world GTFS dataset. The results indicate that the model achieves tighter uncertainty bounds and high interpretability while consistently maintaining high predictive performance across all links.

The promising results of the model suggest it could serve as a potential tool to facilitate the development of public transportation systems. Firstly, it can act as an assessment tool by estimating uncertainty bounds to reflect system operation performance. Transit agencies can use the model to identify and address inefficiencies within their network, improve scheduling accuracy, and optimize resource allocation. Moreover, by having a reliable measure of performance variability, agencies can better communicate service reliability to the public and make informed decisions about infrastructure investments. Secondly, the model’s capacity to accurately predict travel times in real-time, along with its provision of uncertainty bounds, renders it an ideal tool for trip planning. For instance, by providing ranges for travel times to each stop based on the bus's current location, travelers can optimize their boarding times by aiming to arrive before the lower bound of the range. For practitioners, such as urban planners and transit operators, the model provides a valuable tool for scenario analysis and service planning. They can use it as a transit network digital twin \citep{deng2021digital}, which can conduct simulations and help develop control strategies in real-time. This insight helps in designing more resilient and efficient public transport systems, ensuring that service adjustments meet the needs of a growing urban population. Furthermore, the data-driven nature of the model supports evidence-based decision-making, allowing practitioners to justify service changes or infrastructure projects to stakeholders with concrete predictions and analyses.

The study has certain limitations. First, Stop dwell time may be influenced by various factors, such as on-board and off-board activities, the number of passengers using wheelchairs, and other considerations. Due to the unavailability of hourly ridership data, these factors are not currently incorporated into the methodology. Additionally, intersection waiting and passing times are modeled using a simple distribution, which may not fully capture the complexities of queuing behavior. Future work will investigate the influential factors and complex behaviors associated with stop dwell time and intersection waiting and passing time to provide a more comprehensive understanding of bus travel dynamics. Another limitation of this study is that the Markovian model currently only predicts travel times to stops within the same bus route. Future research will aim to extend this analysis from individual routes to the entire network level, thereby improving the understanding of operational performance across the transit network and assisting travelers in planning their journeys throughout the public transportation system.

\section{Acknowledgments}
This research was supported by the U.S. Department of Transportation (USDOT) through the Center for Equitable Transit-Oriented Communities (CETOC) (Grant No. 69A3552348337). Any opinions, findings, and conclusions or recommendations expressed in this material are those of the authors and do not necessarily reflect the views of USDOT. The authors acknowledged using ChatGPT to check grammar errors and improve languages. The authors reviewed and edited the content as needed and take full responsibility for the content. 

%% If you have bibdatabase file and want bibtex to generate the
%% bibitems, please use
%%
\bibliographystyle{elsarticle-harv} 
\bibliography{cas-refs}

\clearpage
\renewcommand{\thesection}{\Alph{section}}
\begin{appendix}
\section*{Appendix A: Baseline Models}
\label{appendix:A}
\renewcommand{\thesubsection}{A.\arabic{subsection}}  % Format subsections
\paragraph{\textbf{Historical Mean (HM): }}The Historical Mean method predicts specific time values by calculating the average of all observed times within the training dataset. The 95\% confidence interval derived from the historical mean method corresponds to the 2.5\% to 97.5\% quantiles of the travel times within the training set.  
\paragraph{\textbf{Linear Regression (LR): }}The model identifies the relationships between bus running time and various factors by representing the running time as a linear combination of these factors and an error term. The error terms are assumed to be normally distributed with a constant variance. 
\paragraph{\textbf{Decision Tree (DT): }}The Decision Tree regression model is applied to predict bus running time by recursively partitioning the training data into subsets that minimize the impurity in the running time within each resulting partition, until a stopping criterion is met, with each final subset forming a leaf node. To make predictions, the model evaluates the factors of a new observation against the decision tree structure to determine its corresponding leaf. The predicted running time is then the mean of the running times of training instances in that leaf. The decision tree method, as a standalone algorithm, does not inherently provide uncertainty bounds for its predictions.
\paragraph{\textbf{Multi-Layer Perceptron (MLP): }}The model is a type of artificial neural network to predict bus running time. It operates through a series of interconnected layers: an input layer, one or more hidden layers, and an output layer. Each node (or neuron) in a layer is connected to nodes in adjacent layers, with each connection assigned a weight that adjusts during training. This architecture allows the MLP model to learn and capture complex, non-linear relationships between the input covariates and the bus running time. The MLP is a deterministic method and also does not offer uncertainty bounds for its predictions.

\section*{Appendix B: Comparative Metrics for Model Evaluation}
\label{appendix:B}
To evaluate the point estimation performance of the proposed model, Root Mean Square Error (RMSE) and Mean Absolute Error (MAE)  were employed as evaluation metrics. RMSE is defined as:
\begin{equation}
RMSE = \sqrt{\frac{1}{n_t}\sum_{i=1}^{n_t}( t_{i,i+1}^k - \hat{t}_{i,i+1}^k)^2.}
\end{equation}
MAE is defined as:
\begin{equation}
MAE = \frac{1}{n_t}\sum_{i=1}^{n_t}\left | t_{i,i+1}^k - \hat{t}_{i,i+1}^k \right |.
\end{equation}
Here, $n_t$ represents the total number of observations in the test set. $t_{i,i+1}^k$ denotes the observed value of link travel time between stop $i$ and $i+1$ for the $k$th observation. $\hat{t}_{i,i+1}^k$ is the corresponding predicted value of the time component for the same observation. The index $k$ ranges from 1 to $n_t$.

Additionally, the width of the confidence intervals in the test set, was compared to evaluate the proposed method against those calculated using the baseline methods.
\end{appendix}

%% else use the following coding to input the bibitems directly in the
%% TeX file.

% \begin{thebibliography}{00}

% %% \bibitem{label}
% %% Text of bibliographic item

% \bibitem{}

% \end{thebibliography}
\end{document}